\documentclass{article}

\usepackage{url}
\usepackage{setspace}
\usepackage{amssymb}
\usepackage{amsmath}
\usepackage{graphicx}

\newtheorem{lemma}{Lemma}
\newtheorem{theorem}{Theorem}
\newtheorem{definition}{Definition}

\newenvironment{proof}{\noindent\textit{Proof.}\ }{$\Box$ \\[-0.02\textheight]}

\title{\vspace*{-0.065\textheight}
\Large\bf The Rank Convergence of HITS Can Be Slow}
\author{Enoch Peserico and Luca Pretto
\footnote{\texttt{\{enoch, pretto\}@dei.unipd.it} - Dipartimento
di Ingegneria dell'Informazione, Universit{\`a} di Padova, Italy.
The first author was supported in part by MIUR under PRIN
Mainstream and by EU under Integr. Proj. AEOLUS (IP-FP6-015964)
The second author was supported in part by Proj. SAPIR within EU
Comm. IST Proj. (Contr. IST-045128).}}

\date{}

\begin{document}

\maketitle

\vspace{0.04\textheight}
\begin{minipage}[c]{0.91\textwidth}
\normalsize{\textbf {Abstract.} We prove that HITS, to ``get
right'' $h$ of the top $k$ ranked nodes of an $N\geq 2k$ node
graph, can require $h^{\Omega(N\frac{h}{k})}$ iterations (i.e. a
substantial $\Omega(N\frac{h\log h}{k})$ matrix multiplications
even with a ``squaring trick''). Our proof requires no algebraic
tools and is entirely self-contained.
}
\end{minipage}
\vspace{0.035\textheight}
\section{\large HITS}
\label{sec:hits}
Kleinberg's celebrated HITS algorithm~\cite{K99} ranks the nodes of a
generic graph in order of ``importance'' based solely on the graph's
topology. Originally proposed to rank web pages
in order of {\em authority} (and still the basis of some search
engines such as Ask~\cite{ask}),
it has been adapted to many different application domains,
such as topic distillation~\cite{CDGKRRT98}, word
stemming~\cite{BFM02b}, automatic synonym extraction in a
dictionary~\cite{B04}, item selection~\cite{Wang&2002}, and author
ranking in question answer portals~\cite{Jurczyk&2007}
(to name just a few - see also
~\cite{Kurland&2006,Mizzaro&2007,Oyama&2006,Kimelfeld&2007}).

The original version of HITS works as follows. In response to a
query, a search engine retrieves a set of nodes of the web graph
on the basis of pure textual analysis; for each such node it also
retrieves all nodes pointed by it, and up to $d$ nodes pointing to
it. Then HITS associates an authority score $a_{i}$ (as
well as a \emph{hub} score $h_{i}$) to each node $v_i$ of this
\emph{base set}, and iteratively updates these scores according to
the formulas:
\begin{equation}
h_{i}^{(0)} = 1 \qquad
a_{i}^{(k)} = \sum_{v_j \rightarrow v_i}h_{j}^{(k-1)} \qquad
h_{i}^{(k)} = \sum_{v_i \rightarrow v_j}a_{j}^{(k)}, \qquad
k=1,2,\ldots \label{equa:for1}
\end{equation}
where $v \rightarrow u$ denotes that $v$ points to $u$. At each
step the authority and hub vector of scores are normalized in
$\|\cdot\|_2$.

Intuitively, HITS places a pebble on each node of the \emph{base
set} graph. At odd timesteps, each pebble on node $v$ sires a
pebble on every node $u$ such that $v\rightarrow u$, and at even
timesteps each pebble on node $v$ sires a pebble on every node $u$
such that $u\rightarrow v$ (a pebble is removed upon siring its
children). Then, without normalization, $a^{(k)}_i$ equals the
number of pebbles on $v_i$ at time $2k-1$ and $h^{(k)}_i$ that at
time $2k$.

\section{\large Convergence in Score vs.\ Convergence in Rank}
\label{sec:convergence}

HITS essentially computes a dominant eigenvector of $A^TA$, where
$A$ is the adjacency matrix of the base set, using the power
method~\cite{GVL96} - thus, the convergence rate of the hub and
authority score vectors are well known~\cite{AP05}. Nevertheless,
what is often really important~\cite{Peserico&2007} is the time
taken by HITS to {\em converge in rank}: intuitively after how
many iterations nodes no longer change their relative rank. A
formalization of this intuition is more challenging than it might
appear~\cite{Peserico&2007}. For the purposes of this paper we
define convergence in rank as follows:

\begin{definition}
\label{def:convergence} Consider an iterative algorithm
ALG providing at every iteration $t\geq 0$ a score vector
$\mathbf{v}^t=[v^t_1,\dots,v^t_N]$ for the $N$ nodes
$v_1,\dots,v_N$ of a graph; and let the set of the (weakly) top
$k$ nodes at step $t$ be $T^t_k=\{v_i:|\{v_j:v^t_j>v^t_i\}|<k\}$.
Then ALG {\em converges on $h$ of the top $k$ ranks in $\tau$
steps} if $|\bigcap_{t=\tau}^\infty T^t_k|\geq h$.
\end{definition}

In other words, an algorithm converges on $h$ of the top $k$ ranks in
$\tau$ steps if after $\tau$ steps it already ``gets right'' at
least $h$ of the $k$ (eventually) top ranked elements. This
definition is closely related to that of convergence in the {\em
intersection metric}~\cite{FKS03,Peserico&2007} for the top $k$
positions; \cite{Peserico&2007} provides a more thorough
discussion of its relationship to other popular metrics such as
Kendall's $\tau$, Cramer-von Mises' $W^2$, or Kolmogorov-Smirnov's
$D$.

We prove the first non-trivial lower bound on the iterations HITS
requires to converge in rank. All previous rank convergence
studies save \cite{Peserico&2007} are experimental, and none
investigates HITS (focusing instead on
PageRank~\cite{H99,KHMG03,Peserico&2007}).

\section{\large HITS Can Converge Slowly in Rank}
\label{sect:theorem}
Informally, we prove that HITS on an $N$ node graph can take
$h^{\Omega(N\frac{h}{k})}$ steps to ``get right'' $h$ of the top
$k$ elements. This effectively means $\Omega(\frac{N h\log h}{k})$
matrix multiplications using the standard ``squaring trick'' that computes
the $p^{th}$ power of a matrix $M$ by first computing the matrices
$M^2,M^4,\ldots,M^{2^{\lfloor\log p\rfloor}}$. More formally, we
devote the rest of this section to the proof of:

\begin{theorem}
\label{thm:general} For all $h$ and $k$ such that $k>h>5$, and all
odd $n\geq\max(3,\frac{k-h+2}{2})$, there is an undirected graph
$\Gamma_{h,k,n}$ of
$N=\lceil\frac{k-2}{h-3}\rceil(2n+h-3)+1\approx\frac{2nk}{h}+k$
vertices on which HITS requires more than $\bar
t=\frac{3\ln(7/6)}{4e}(\frac{h-3}{2})^{\frac{n-1}{2}}=
h^{\Omega(n)}$ steps to converge on $h$ of the top $k$ ranks (and
the last term is $h^{\Omega(N\frac{h}{k})}$ for $N\geq 2k$).
\end{theorem}

$\Gamma_{h,k,n}$ (Fig.~\ref{fig:gamma}) is formed by a subgraph
$\bar\Gamma_{m,n}$ (with $m=h-3$) and
$\ell=\lceil\frac{k-h+1}{h-3}\rceil$ isomorphic subgraphs
$\Gamma^1_{m,n},\dots,\Gamma^\ell_{m,n}$. $\bar\Gamma_{m,n}$ has
$2n+m+1$ vertices, $v_{-n},\ldots,v_0,\ldots,v_{n+m}$. The $2n+1$
vertices $v_{-n},\ldots,v_0,\ldots,v_{n}$ form a chain, with $v_i$
connected to $v_{i-1}$. The first and last vertices of the chain,
$v_{-n}$ and $v_n$, are also connected to each of the $m$ vertices
$v_{n+1},\ldots,v_{n+m}$. $\Gamma^1_{m,n}$ has $2n+m$ vertices,
$u_{-n},\ldots,u_{-1},u_1,\ldots,u_{n+m}$, and is almost
isomorphic to $\bar\Gamma_{m,n}$: $u_i$ is connected to $u_j$ if
and only if $v_i$ is connected to $v_j$. The only difference is
that $u_0$ is missing.
%%%%%%%%%%%%%%%%%%%%%%%%%%%%%%%%%%%%%%%%%%%%%%%%%%%%%%%%%%%%%%%%
\vspace*{.01\textheight}
\begin{figure}[h]
\begin{center}
\resizebox{.95\textwidth}{!}{\includegraphics{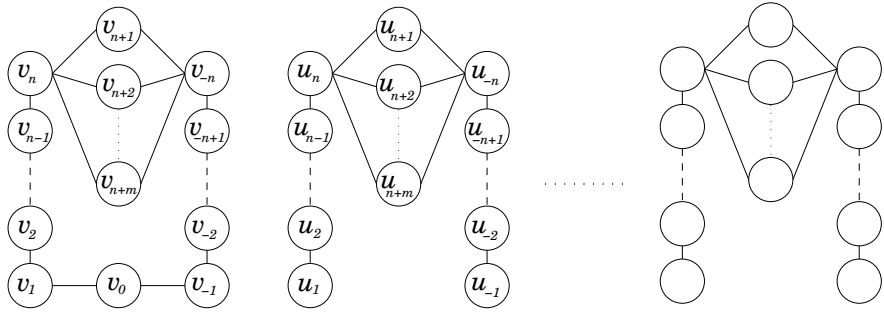}}
\caption{\small The graph $\Gamma_{h,k,n}$ is formed by the
subgraph $\bar\Gamma_{m,n}$ (first left) and $\ell$ subgraphs
isomorphic to the subgraph $\Gamma^1_{m,n}$ (second left), with
$m\approx h$ and $\ell\approx\frac{k}{h}$.} \label{fig:gamma}
\end{center}
\end{figure}
%%%%%%%%%%%%%%%%%%%%%%%%%%%%%%%%%%%%%%%%%%%%%%%%%%%%%%%%%%%%%%%%

\noindent The proof of the theorem proceeds as follows. After
introducing some notation, Lemma~\ref{lem:localgrowth} bounds the
growth rate of the number of pebbles on $v_{i}$ and ${u_i}$ as a
function of $i$. Lemma~\ref{lem:localgrowth} allows us prove, in
Lemma~\ref{lem:fastgrow} that eventually only a vanishing fraction
of all pebbles resides outside $\bar\Gamma_{m,n}$, and in
Lemma~\ref{lem:slowgrow} that $v_{n+1}$ acquires pebbles only
minimally faster than $u_{n+1}$. We then prove the theorem showing
that $\bar\Gamma_{m,n}$ eventually holds all the top $k$ nodes,
but $u_{n+1},\dots,u_{n+m}$ and the corresponding nodes in
$\Gamma^2_{m,n},\dots,\Gamma^\ell_{m,n}$ for $t\leq \bar t$ still
outrank $v_{n-1},\dots,v_0,\dots,v_{-n+1}$ (and thus at least
$\ell m>k-h$ of the top $k$ nodes lie outside $\bar\Gamma_{m,n}$).

Denote by $v^t$ the number of descendants at time $t$ of a pebble
present at time $0$ on $v$ - which is also equal to the number of
pebbles present on $v$ after a total of $t$ timesteps, since both
quantities are described by the recursive equation
$v^{t+1}=\sum_{u\leftrightarrow v}u^t$ with $v^0=1$.

Also, mark with a timestamp $\tau$ any pebble present at time
$\tau$ on $v_0$ and any pebble not on $v_0$ whose most recent
ancestor on $v_0$ was present at time $\tau$. Note that the number
of unmarked pebbles present at any given time on $v_i$ (for any
$i$) is equal to the total number of pebbles present at that time
on $u_i$; and, more generally, it is straightforward to verify by
induction on $t-\tau$ that any pebble present on a vertex $u_i$ at
time $\tau$ has, on any vertex $u_j$ and at any time $t\geq \tau$,
a number of descendants equal to the number of descendants not
marked after $\tau$ that any pebble present on a vertex $v_i$ at
time $\tau$ has on $v_j$ at time $t$.

\begin{lemma}
\label{lem:localgrowth} For any $t\geq 0$, and any $i,j$ such that
$(i \equiv j)$ mod $2$ and $0 \leq i < j \leq n+1$, we have $1
\leq \frac{v_i^{t+1}}{v_i^t} \leq \frac{v_j^{t+1}}{v_j^t} \leq
m+1$ and similarly (if $i>0$) $1 \leq \frac{u_i^{t+1}}{u_i^t} \leq
\frac{u_j^{t+1}}{u_j^t} \leq m+1$.
\end{lemma}

\begin{proof}
We prove that $1 \leq \frac{v_i^{t+1}}{v_i^t} \leq
\frac{v_j^{t+1}}{v_j^t} \leq m+1$ by induction on $t$. The base
case $t=0$ is easily verified. $\frac{v_h^{t+1}}{v_h^t}=
\frac{\sum_{v_{h'}\leftrightarrow v_h}v_{h'}^t}
{\sum_{v_{h'}\leftrightarrow v_h}v_{h'}^{t-1}}$ is a weighted
average (with positive weights) of all ratios
$\frac{v_{h'}^t}{v_{h'}^{t-1}}$. By inductive hypothesis, $1 \leq
\min_{v_{i'}\leftrightarrow v_i} \frac{v_{i'}^t}{v_{i'}^{t-1}}
\leq \frac{v_i^{t+1}}{v_i^t} \leq \max_{v_{i'}\leftrightarrow v_i}
\frac{v_{i'}^t}{v_{i'}^{t-1}} = \frac{v_{i+1}^t}{v_{i+1}^{t-1}}
\leq \frac{v_{j-1}^t}{v_{j-1}^{t-1}} = \min_{v_{j'}\leftrightarrow
v_j} \frac{v_{j'}^t}{v_{j'}^{t-1}} \leq \frac{v_j^{t+1}}{v_j^t}
\leq \max_{v_{j'}\leftrightarrow v_j}
\frac{v_{j'}^t}{v_{j'}^{t-1}} \leq m+1$. The proof that $1 \leq
\frac{u_i^{t+1}}{u_i^t} \leq \frac{u_j^{t+1}}{u_j^t} \leq m+1$ is
identical.
\end{proof}

\begin{lemma}
\label{lem:fastgrow}
$\forall i>0$ $\underset{t\rightarrow\infty}{\lim}(u^t_i/v^t_i)=0$.
\end{lemma}

\begin{proof}
For all $i\geq 0$ consider an unmarked pebble $p_i$ present at
time $t$ on vertex $v_i$. At some time $t+\tau$, with $\tau \leq
2n+2$, $v_{n+1}$ holds at least one marked descendant $p'_i$ of
$p_i$; by virtue of Lemma \ref{lem:localgrowth}, $p'_i$ thereafter
always has at least as many descendants as any of the other at
most $(m+1)^\tau$ descendants of $p_i$ present at time $t+\tau$.
Then, every $2n+2$ timesteps, the fraction of unmarked descendants
of an unmarked pebble drops by a factor at least
$1-(m+1)^{-(2n+2)}$.
\end{proof}

\begin{lemma}
\label{lem:slowgrow} $\frac{v^{t+1}_{n+1}/v^{t}_{n+1}}
{u^{t+1}_{n+1}/u^{t}_{n+1}} \leq
1+\frac{m+1}{m}\frac{v^{t}_{0}}{u^{t}_{n+1}}$.
\end{lemma}

\begin{proof}
Denote by $D^t_\tau$ the number of descendants at time $t$ of a
pebble initially on $v_{n+1}$ whose timestamp is $\tau$, and with
$D^t_u$ the number of those descendants yet unmarked. Since all
pebbles whose timestamp is $\tau$ descend from pebbles present in
$v_0$ at time $\tau$, Lemma~\ref{lem:localgrowth} guarantees that
the growth rates of $D^t_{(\cdot)}$ satisfy
$\frac{D^{t+1}_u}{D^t_u} \geq \frac{D^{t+1}_\tau}{D^{t}_\tau}$ for
any $\tau$ for which $D^{(\cdot)}_\tau \neq 0$. Thus,
$\frac{v^{t+1}_{n+1}}{v^{t}_{n+1}}=
\frac{D^{t+1}_u+D^{t+1}_0+\dots+D^{t+1}_t}{D^{t}_u+D^{t}_0+\dots+D^{t}_t}+
\frac{D^{t+1}_{t+1}}{D^{t}_u+D^{t}_0+\dots+D^{t}_t}\leq
\frac{D^{t+1}_u}{D^t_u}+\frac{v^{t+1}_{0}/m}{v^{t}_{n+1}}=
\frac{u^{t+1}_{n+1}}{u^t_{n+1}} +
\frac{1}{m}\frac{v^{t+1}_{0}}{v^{t}_{n+1}}$; and therefore
$\frac{v^{t+1}_{n+1}/v^{t}_{n+1}} {u^{t+1}_{n+1}/u^{t}_{n+1}}\leq
1+ \frac{1}{m}\frac{v^{t+1}_{0}}{v^{t}_{n+1}}
\cdot\frac{u^{t}_{n+1}}{u^{t+1}_{n+1}}\leq
1+\frac{1}{m}\frac{v^{t+1}_{0}}{u^{t+1}_{n+1}} \leq
1+\frac{m+1}{m}\frac{v^{t}_{0}}{u^{t}_{n+1}}$.
\end{proof}

\vspace{0.04\textheight}

 \noindent We can now prove Theorem \ref{thm:general}. By
Lemma \ref{lem:fastgrow}, $\underset{t\rightarrow\infty}{\lim}
(u^t/\sum_{w\in\Gamma_{h,k,n}}w^t) =0~\forall u\notin
\bar\Gamma_{m,n}$; whereas $\forall t,~0\leq i\leq n+1$, by Lemma
\ref{lem:localgrowth}
$v_i^t\geq\frac{1}{(m+2n+1)(m+1)^{n+1}}\sum_{v\in\bar\Gamma_{m,n}}v^t$.
Thus, eventually the top $(m+2n+1)\geq k$ ranked nodes all belong
to $\bar\Gamma_{m,n}$.

We complete the proof showing that, for
$n-1\leq t\leq \bar t=\frac{3\ln(7/6)}{4e}(\frac{m}{2})^{\frac{n-1}{2}}
= m^{\Omega(n)}$, we have
$\frac{\max(v_{n-1}^t,~v_{n-2}^t)}{u_{n+1}^t}\leq\frac{7}{8}$
and thus by Lemma~\ref{lem:localgrowth}
at least $\ell m$ elements outside $\bar\Gamma_{m,n}$ are
among the top $k$ ranked nodes.
Note that
$\frac{v_{n-1}^3}{v_{n+1}^3}=\frac{2m+6}{4m+4}$; that
$\frac{v_{n-2}^3}{v_{n}^3}=\frac{m+7}{2m^2+3m+3}$; and that
$\frac{v_n^t}{v_{n+1}^t}\leq\frac{(m+1)v_n^{t-1}}{2v_n^{t-1}}=\frac{m+1}{2}$.
Then, by Lemma \ref{lem:localgrowth},
$\frac{\max(v_{n-1}^t,~v_{n-2}^t)}{u_{n+1}^t}\leq
\frac{v_{n+1}^t}{u_{n+1}^t}\cdot\max(\frac{v_{n-1}^t}{v_{n+1}^t},
\frac{v_{n}^t}{v_{n+1}^t}\cdot\frac{v_{n-2}^t}{v_n^t})\leq
\frac{v_{n+1}^t}{u_{n+1}^t}\cdot\frac{6}{8}$ for $m\geq 3$. All is left
to prove is that $\frac{v_{n+1}^t}{u_{n+1}^t}\leq\frac{7}{6}$ for
$n-1\leq t\leq \bar t$.

We first prove that, for $n-1\leq t\leq\bar t$,
$\frac{v^{t}_{0}}{u^{t}_{n+1}}\leq
e(\frac{2}{m})^{\frac{n-1}{2}}$. It is straightforward that
$v^{n-1}_0=2^{n-1}$ and $u^{n-1}_{n+1}=v^{n-1}_{n+1}\geq
(2m)^{\frac{n-1}{2}}$. Then, $\forall t$ such that
$\frac{v^t_0}{u^t_{n+1}} \leq e(\frac{2}{m})^{\frac{n-1}{2}} =
\frac{3}{4\tilde{\tau}}$, we have $\frac{v^{t+1}_0 /
v^t_0}{u^{t+1}_{n+1} / u^t_{n+1}}= \frac{v^{t+1}_0 /
v^t_0}{v^{t+1}_{n+1} / v^t_{n+1}} \cdot \frac{v^{t+1}_{n+1} /
v^t_{n+1}} {u^{t+1}_{n+1} / u^t_{n+1}} \leq 1\cdot
(1+\frac{(m+1)}{m} \frac{3}{4\tilde{\tau}})\leq
1+\frac{1}{\tilde{\tau}}$ and it takes at least $\tilde{\tau}$
timesteps for $\frac{v^t_0}{u^t_{n+1}}$ to grow by a factor $e$ to
$e(\frac{2}{m})^{\frac{n-1}{2}}$. Thus,
$\frac{v_{n+1}^t}{u_{n+1}^t}= \frac{v_{n+1}^{n-1}}{u_{n+1}^{n-1}}
\cdot\Pi_{\tau=n-1}^{t-1} \frac{v^{\tau+1}_{n+1}/v^{\tau}_{n+1}}
{u^{\tau+1}_{n+1}/u^{\tau}_{n+1}}\leq
1\cdot(1+\frac{1}{\tilde{\tau}})^{\tilde{\tau}\ln(7/6)}\leq
\frac{7}{6}$.

\section{\large Conclusions and Open Problems}
This paper presents a self-contained proof that HITS might require
$h^{\Omega(N\frac{h}{k})}$ iterations to ``get right'' $h$ of the
top $k$ nodes of an $N\geq 2k$ node graph. This translates into
$\Omega(N\frac{h\log h}{k})$ matrix multiplications even using a
``squaring trick''- a substantial load when HITS must be used
on-line on large graphs (e.g. in web search engines).

We conjecture that $g^{\Omega(N\frac{h}{k})}$ is a tight worst
case bound on the iterations required by HITS to converge in rank
on $h$ of the top $k$ ranked nodes of an $N\geq 2k$ node graph of
maximum degree $g$. This is slightly more (for $h$ subpolynomial
in $N$) than the lower bound presented here.

\small \setstretch{0.5}

%\bibliographystyle{abbrv}
%\small{\bibliography{slowhits}}

\end{document}